\documentclass[runningheads,envcountsect,envcountsame]{llncs}

\usepackage[T2A,T1]{fontenc}
\usepackage[utf8]{inputenc}
\usepackage[russian,english]{babel}

\usepackage{amssymb}
\usepackage{amsmath}
\usepackage{graphicx}
\usepackage[dvipsnames]{xcolor}
\usepackage{breakcites}
\usepackage{verbatim}
\usepackage{wrapfig}
\usepackage{subfig}

\providecolor{DarkGreen}{rgb}{0,.392,0}
\providecolor{DarkBlue}{rgb}{0,0,.545}
\usepackage[colorlinks=true,citecolor=DarkGreen,linkcolor=DarkBlue]{hyperref}
\usepackage[capitalise,noabbrev]{cleveref}

\usepackage[ruled,linesnumbered]{algorithm2e}
\SetKwInput{KwInput}{Input}
\SetKwInput{KwOutput}{Output}

\usepackage{todonotes}

\usepackage{navigator}

\usepackage{pgfplots, pgfplotstable}
\usepackage{tikz}
\usetikzlibrary{arrows, chains, matrix, positioning, scopes, patterns, shapes}
\usetikzlibrary{decorations.pathreplacing}

\usetikzlibrary{arrows}
\usepackage{multirow}

\usepackage{braket}
\usepackage{bbm}
\usepackage{thm-restate}
\usepackage{makecell}
\newcommand{\RR}{\mathbb{R}}

\newcommand{\NN}{\mathbb{N}}

\newcommand{\EE}{\mathbb{E}}

\newcommand{\Var}{\operatorname{Var}}
\newcommand{\renta}[2]{\operatorname{H}_{#1}(#2)}

\newcommand{\softO}{\widetilde{\mathcal{O}}}

\newcommand{\softTheta}{\widetilde{\Theta}}

\newcommand{\dist}{\ensuremath \mathcal{D}}
\newcommand{\keys}{\mathcal{K}}
\newcommand{\oracle}{\mathbbm{1}_k}
\newcommand{\multiOracle}{\mathbbm{1}_{k_1},\ldots,\mathbbm{1}_{k_m}}

\newcommand{\TC}{T_C(\dist)}
\newcommand{\TQ}{T_Q(\dist)}
\newcommand{\coreset}{\mathcal{C}_{\chi}^{n,\delta}}
\newcommand{\typical}{\mathcal{T}_{\chi}^{n,\delta}}

\newcommand{\ent}{\ensuremath \operatorname{H}}
\newcommand{\keyenum}{{\normalfont \textsc{GetKey}_{\dist}}}
\newcommand{\fullguess}{\normalfont \textsc{KeyGuess}}
\newcommand{\multiguess}{\normalfont \textsc{MultiKeyGuess}}
\newcommand{\abortguess}{\normalfont \textsc{AbortedKeyGuess}}
\newcommand{\qguess}{\normalfont \textsc{QuantumMultiKeyGuess}}

\newcommand{\B}{\mathcal{B}}

\newcommand{\T}{\mathcal{T}}
\newcommand{\D}{\mathcal{D}}
\newcommand{\Z}{\mathcal{Z}}
\newcommand{\K}{\mathcal{K}}

\Crefname{algocf}{Algorithm}{Algorithms}
\Crefname{heuristic}{Heuristic}{Heuristics}
\Crefname{estimate}{Heuristic Claim}{Heuristic Claims}

\definecolor{midnightblue}{RGB}{0, 73, 122}
\definecolor{tealgreen}{RGB}{0, 135, 120}
\definecolor{goldenrod}{RGB}{255, 190, 0}
\definecolor{coralred}{RGB}{204, 85, 0}
\definecolor{plum}{RGB}{153, 51, 102}
\definecolor{olivegreen}{RGB}{107, 142, 35}
\definecolor{royalblue}{RGB}{25, 100, 190}
\definecolor{forestgreen}{RGB}{0, 95, 70}
\definecolor{amber}{RGB}{255, 150, 0}
\definecolor{brickred}{RGB}{165, 40, 40}
\definecolor{slategray}{RGB}{112, 128, 144}

\definecolor{lightgreen}{rgb}{0.86, 0.93, 0.78}
\definecolor{bordergreen}{rgb}{0.55, 0.76, 0.74}
\definecolor{lightblue}{rgb}{0.70, 0.90, 0.99}
\definecolor{borderblue}{rgb}{0.01, 0.66, 0.96}
\definecolor{lightamber}{rgb}{1, 0.93, 0.70}
\definecolor{borderamber}{rgb}{1, 0.76, 0.03}
\definecolor{lightcolor4}{rgb}{ 0.93, 0.70, 1}
\definecolor{bordercolor4}{rgb}{0.76, 0.03, 1}
\definecolor{lightcolor5}{rgb}{0.78,0.86,0.93}
\definecolor{bordercolor5}{rgb}{0.74,0.55,0.76}

\RequirePackage{totcount}
\RequirePackage{color}
\newtotcounter{notecount}

\newcommand{\anote}[1]{\stepcounter{notecount}\todo[inline,bordercolor=bordercolor4,linecolor=bordercolor4,color=lightcolor4]{\footnotesize{\sc \bf Alex:} #1}{}}
\newcommand{\jnote}[1]{\stepcounter{notecount}\todo[inline,bordercolor=borderamber,linecolor=borderamber,color=lightamber]{\footnotesize{\sc \bf Julian:} #1}{}}

\def\addlegendimage{\csname pgfplots@addlegendimage\endcsname}

\title{Super-Quadratic Quantum Speed-ups\\ and Guessing Many Likely Keys}
\author{
		Timo Glaser
		\href{https://orcid.org/0009-0001-8970-3153}
		{\protect\includegraphics[height=\fontcharht\font`B]{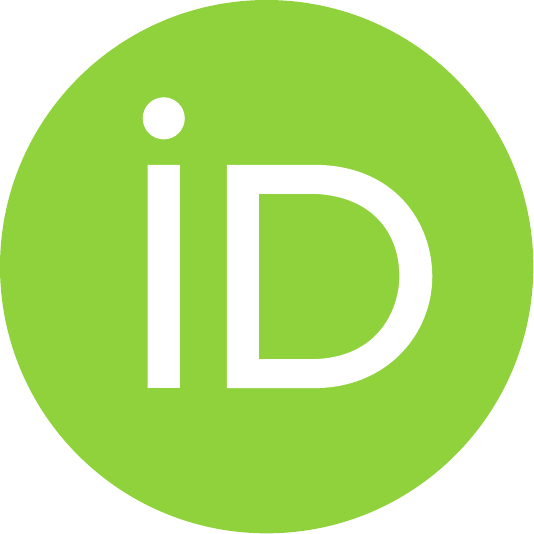}}
		\and
		Alexander May
		\href{https://orcid.org/0000-0001-5965-5675}
		{\protect\includegraphics[height=\fontcharht\font`B]{orcid.pdf}}
		\and
		Julian Nowakowski
		\href{https://orcid.org/0000-0003-3066-0133}
	{\protect\includegraphics[height=\fontcharht\font`B]{orcid.pdf}}
	}

\institute{
	Ruhr-University Bochum, Bochum, Germany\\ \email{\{timo.glaser,alex.may,julian.nowakowski\}@rub.de}
}

\begin{document}

\maketitle

\begin{abstract}
    We study the fundamental problem of guessing cryptographic keys, drawn from some non-uniform probability distribution $\D$, as e.g. in LPN, LWE or for passwords. The optimal classical algorithm enumerates keys in decreasing order of likelihood. The optimal quantum algorithm, due to Montanaro (2011), is a sophisticated Grover search.

    We give the first tight analysis for Montanaro's algorithm, showing that its runtime is $2^{\ent_{2/3}(\dist)/2}$, where $\ent_{\alpha}(\cdot)$ denotes Renyi entropy with parameter~$\alpha$.
    Interestingly, this is a direct consequence of an information theoretic result called Arikan's Inequality (1996) -- which has so far been missed in the cryptographic community -- that tightly bounds the runtime of classical key guessing by $2^{\ent_{1/2}(\dist)}$.
    
    Since $\ent_{2/3}(\dist) < \ent_{1/2}(\dist)$ for every non-uniform distribution $\dist$, we thus obtain a \emph{super-quadratic} quantum speed-up $s>2$ over classical key guessing.
     To give some numerical examples, for the binomial distribution used in Kyber, and for a typical password distribution, we obtain quantum speed-up $s>2.04$. For the $n$-fold Bernoulli distribution with parameter $p=0.1$ as in LPN, we obtain $s > 2.27$.  For small error LPN with $p=\Theta(n^{-1/2})$ as in Alekhnovich encryption, we even achieve \emph{unbounded} quantum speedup $s = \Omega(n^{1/12})$.

     As another main result, we provide the first thorough analysis of guessing in a multi-key setting.
     Specifically, we consider the task of attacking many keys sampled independently from some distribution $\dist$, and aim to guess a fraction of them.
     For product distributions $\D = \chi^n$, we show that any constant fraction of keys can be guessed within $2^{\ent(\dist)}$ classically and $2 ^{\ent(\dist)/2}$ quantumly per key, where $\ent(\chi)$ denotes Shannon entropy.
     In contrast, Arikan's Inequality implies that guessing a single key costs $2^{\ent_{1/2}(\dist)}$ classically and $2^{\ent_{2/3}(\dist)/2}$ quantumly.
     Since\linebreak $\ent(\dist) < \ent_{2/3}(\dist) < \ent_{1/2}(\dist)$, this shows that in a multi-key setting the guessing cost per key is substantially smaller than in a single-key setting, both classically and quantumly.

\end{abstract}

\section{Introduction}


\subsubsection{Key guessing.}
The most fundamental problem in cryptanalysis is the \emph{key guessing problem}.
Formally, let $k \in \keys$ be a cryptographic key sampled from some key space $\keys$, and let $\oracle: \keys \to \{0,1\}$ be an oracle with
$$
	\oracle(x) := \begin{cases}
		1 & \text{if } x = k,\\
		0 & \text{else.} 
	\end{cases}
$$ 
The key guessing problem asks to find $k$, given oracle access to $\oracle(\cdot)$.

The optimal \emph{classical} algorithm for key guessing is brute force search, i.e., iterating over the key space $\keys$ in uniformly random order until one finds~$k$.
On expectation, brute force finds~$k$ after $\Theta(|\keys|)$ steps.
The optimal \emph{quantum} algorithm is Grover search~\cite{Grover96}, which achieves a square root speed-up over the classical algorithm, thus solving the key guessing problem in time $\Theta(\sqrt{|\keys|})$.

\subsubsection{Key guessing from a distribution.}
The situation changes if the attacker knows the probability distribution $\dist$ of the key $k$.
In that case, the optimal classical algorithm for key guessing is to query $\oracle(\cdot)$ on keys $x \in \keys$ in decreasing order of probability.
This approach has expected runtime
$$
	\TC := \sum_{i=1}^{|\keys|} i \cdot p_i,
$$
where $p_i$ is the probability of the $i$-th likeliest key.

An information theoretic inequality by Arikan~\cite{arikan1996inequality} shows
\begin{equation}
    \label{eq:tc-bound}
    \frac{2^{\ent_{1/2}(\dist)}}{1+\log_2|\keys|}  \leq \TC \leq 2^{\ent_{1/2}(\dist)}
\end{equation}
where $\ent_{1/2}(\dist)$ is the \emph{R\'enyi entropy} of the distribution~$\dist$ with parameter $\frac{1}{2}$.
Notice, if $\dist$ is the uniform distribution over $\keys$, then $2^{\ent_{1/2}(\dist)} = |\keys|$.
However, for non-uniform~$\dist$, we have $2^{\ent_{1/2}(\dist)} < |\keys|$.
Hence, for $\dist$ different from the uniform distribution, enumerating keys in decreasing order of likelihood yields a non-trivial speed-up.

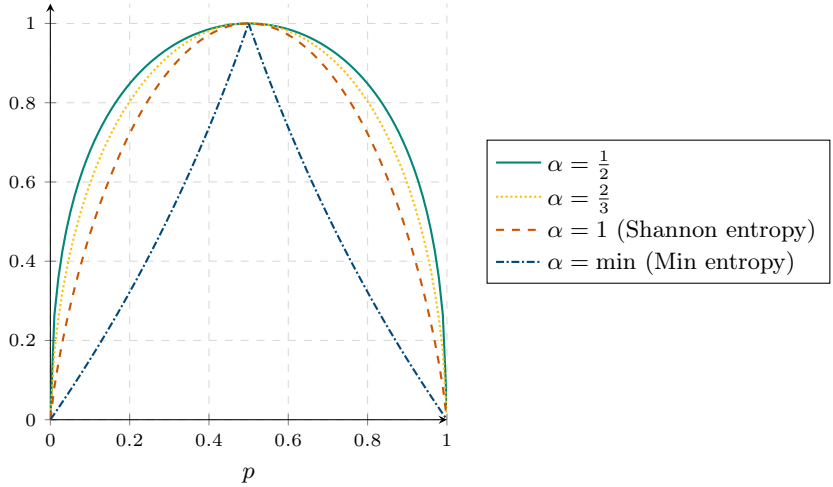
\begin{figure}
    \centering

    \begin{tikzpicture}
            \begin{axis}[
            axis lines=left,
            cycle list name=exotic,
            grid = major,
            grid style = {dashed, gray!30},
            x=150,
            y=150,
            xmin=0,
            xmax=1,
            ymin=0,
            ymax=1.05,
            xlabel = {$p$},
            xticklabel style = {font=\scriptsize},		
            yticklabel style = {font=\scriptsize, 
                /pgf/number format/fixed,
                /pgf/number format/precision=1 },
            legend style = {
             	at = {(1.1,0.5)},
             	anchor = west, 
            },
            legend cell align={left},
            ]
            
            \addplot [domain=0:1,samples=100,thick,tealgreen] {2*ln(sqrt(x)+sqrt(1-x))/ln(2)};
            \addplot [domain=0:1,samples=100,
            thick,densely dotted,goldenrod] {3*ln(x^(2/3)+(1-x)^(2/3))/ln(2)};
            \addplot [domain=0:1,samples=100,thick,dashed,coralred] {-x*ln(x)/ln(2)-(1-x)*ln(1-x)/ln(2)};
            \addplot [domain=0:0.5,samples=100,thick,densely dashdotted,midnightblue] {-ln(1-x)/ln(2)};
            \addplot [domain=0.5:1,samples=100,thick,densely dashdotted,midnightblue] {-ln(x)/ln(2)};

            \addlegendentry{$\alpha = \frac{1}{2}$};
            \addlegendentry{$\alpha = \frac{2}{3}$};
            \addlegendentry{$\alpha = 1$ (Shannon entropy)};
            \addlegendentry{$\alpha = \text{min}$ (Min entropy)};

            \end{axis}
    \end{tikzpicture}
    
    \caption{Rènyi entropy $\ent_{\alpha}(\text{Ber}(p))$ of the Bernoulli distribution $\text{Ber}(p)$ for various $\alpha$.}
    \label{fig:binomial-classic}
\end{figure}

For instance, suppose that $\dist$ is the $n$-fold Bernoulli distribution $\text{Ber}(p)^n$ with parameter $p \in (0,1)$.
That is, $\dist$ samples binary strings from $\{0,1\}^n$ and sets each bit independently to 1 with probability~$p$.
Then
$$
    |\keys| = 2^n, \quad \text{and}\quad 2^{\ent_{1/2}(\dist)} = 2^{\ent_{1/2}(\text{Ber}(p))n}.
$$
As \cref{fig:binomial-classic} shows, we have $\ent_{1/2}(\text{Ber}(p)) < 1$ for every $p \neq \frac{1}{2}$.
Hence, for every constant $p \neq \frac{1}{2}$, enumerating keys in decreasing order of likelihood improves over brute-force by an exponential factor
$$
    \frac{|\keys|}{\TC} = \softTheta\left(\frac{2^n}{2^{\ent_{1/2}(\text{Ber}(p))n}}\right) = 2^{\Omega(n)}.
$$
On the other hand, if $p = \frac{1}{2}$, then $\text{Ber}(p)^n$ is the uniform distribution over $\{0,1\}^n$, and we do not improve over brute force.

Somewhat curiously, Arikan's result seems to be widely unknown in the cryptographic community.
Consider, for instance, the recent line of research on \emph{hybrid attacks} on lattice-based crypto~\cite{matzov,C:DucPul23,EPRINT:AlbShe22,budMar23}.
Hybrid attacks try to recover a part of the secret key via key guessing, and the remaining part via lattice-based techniques.
The cryptographic literature contains the following results for key guessing in hybrid attacks:
\begin{itemize}
	\item Bernstein~\cite{EPRINT:Bernstein23c} uses the loose upper bound $\TC \leq |\keys|$ in the analysis of his hybrid attack.
	Since for any $\dist$ different from the uniform distribution we have $2^{\ent_{1/2}(\dist)} < |\keys|$, Arikan's result shows that Bernstein \emph{overestimates} the runtime of key guessing.
	\item The analysis of MATZOV~\cite{matzov} claims (without proof) that\linebreak $\TC = 2^{\ent(\dist)}$, with $\ent(\dist)$ the \emph{Shannon entropy} of $\dist$.
	Since for non-uniform $\dist$ we have $2^{\ent_{1/2}(\dist)} > 2^{\ent(\dist)}$, it follows from Arikan's result that MATZOV \emph{underestimates} the complexity of key guessing.
	\item Independently of Arikan's result, MATZOV's inaccuracy has been criticized by follow-up works:
	Ducas-Pulles~\cite{C:DucPul23} give counterexample distributions $\dist$ for which MATZOV's claim is wrong.
	Albrecht-Shen~\cite{EPRINT:AlbShe22} try to improve MATZOV's analysis by computing an upper bound on $\TC$ with discrete Gaussian $\dist$.
	Budroni-Mårtensson~\cite{budMar23} give an efficient algorithm that, for certain distributions $\dist$, can compute $\TC$ numerically.
\end{itemize}

\subsubsection{Quantum key guessing from a distribution.}
For the quantum setting, Ambainis and de Wolf~\cite{ambainis2001average} were the first to show super-quadratic speed-ups for non-uniform distributions. 
Montanaro~\cite{montanaro2011quantum} gave an optimal algorithm that, given a distribution $\dist$, solves the key guessing problem in time
$$
	\TQ := \sum_{i=1}^{|\keys|}  \sqrt{i} \cdot p_i.
$$
Let $\TC = \TQ^s$ for some {\em quantum speed-up} $s$. 
For any distribution $\dist$, we have
\begin{align}
\label{eq:montanaro-speedup}
    \TQ \leq \sqrt{\TC}.
\end{align}
Hence, Montanaro's quantum algorithm obtains a quantum speed-up over the optimal classical algorithm of at least 
\[
  s = \frac{\log \TC}{\log \TQ} \geq 2.
\]
We call a speed-up with $s=2$ {\em quadratic}, $s>2$ {\em super-quadratic}, $s=\Theta(1)$ {\em polynomial}, $s=\omega(1)$ {\em super-polynomial}, and with $s=\Theta(n)$ {\em exponential}.

Montanaro showed that there exist probability distributions for which \Cref{eq:montanaro-speedup} is strict, i.e., $\TQ < \sqrt{\TC}$, achieving a super-quadratic speed-up. 
However, these distributions are rather artificial.
For general distributions, the impact of Montanaro's algorithm seemed unclear so far.
Especially, for distributions of cryptographic interest no super-quadratic speed-up $s>2$ is known.

\subsection{Our Contributions}

\subsubsection{Tight Analysis of Montanaro's Quantum Algorithm.}
We observe that Arikan's inequality~\cite{arikan1996inequality} is sufficiently powerful to not only cover classical but also quantum key guessing in the case of distributions $\D$ with finite support $\K$. Namely, a direct application of~\cite{arikan1996inequality} yields
\begin{align}
\label{eq:tq-bound}
    \frac{2^{\frac{\ent_{2/3}(\dist)}{2}}}{\sqrt{1+\log_2|\K|}} \leq \TQ \leq 2^{\frac{\ent_{2/3}(\dist)}{2}}.
\end{align}
Combining~\Cref{eq:tc-bound} and \Cref{eq:tq-bound}, we obtain quantum speed-up
\[
  s := \frac{\log \TC}{\log \TQ} \geq 2 \cdot \frac{\ent_{1/2}(\dist)-\log_2(1+\log_2|\K|)}{\ent_{2/3}(\D)}.
\]
Thus, all distributions $\D$ with $\log\log|\K|=o\left(\ent_{1/2}(\D)\right)$, 
which e.g. holds for any product distribution $\D = \chi^n$ with finite support, 
admit  quantum speed-up 
\[
  s \geq 2 \cdot \frac{\ent_{1/2}(\D)}{\ent_{2/3}(\D)}(1-o(1)).
\]
Since $\ent_{1/2}(\D) > \ent_{2/3}(\D)$ for all but the uniform distribution, this leads asymptotically to a {\em super-quadratic} quantum speed-up.

\subsubsection{Super-Quadratic (and Beyond) Quantum Speed-Ups.}
We study several distributions that are motivated by cryptographic key choices, such as Bernoulli (LPN~\cite{STOC:BluKalWas00,AC:HopBlu01,FOCS:Alekhnovich03}), Ternary (NTRU~\cite{HofPipSil98,CHES:HRSS17}), Binomial (Kyber~\cite{bos2018crystals}), and a Discrete Gaussian (LWE~\cite{STOC:Regev03}). For completeness, we also analyze common distributions that frequently appear in theory such as the Geometric and the Poisson distribution.

We show asymptotically super-quadratic quantum speed-ups for all these distributions, sometimes even {\em unbounded} speed-ups that grow as a function of~$n$.

\subsubsection{Multi-Key Guessing.}
Typically, a large-scale real-world adversary Eve does not target a \emph{single} cryptographic key $k$, but rather a plethora $m$ of collected keys $k_1,\ldots,k_m$ at once, with the goal of recovering a fraction $cm$ of all keys for some $0 \leq c \leq 1$.
Arikan's Inequality tightly bounds the required time for guessing a single key \emph{with success probability 1} by \emph{Rényi} entropy $\ent_{1/2}(\dist)$ classically and $\ent_{2/3}(\dist)$ quantumly.
However, the inequality does not apply in the multi-key setting, where a success probability of roughly $c$ per key suffices.

We show, for any constant $c$ and keys drawn from a product distribution $\dist = \chi^n$, that the classical runtime for guessing $cm$ keys is upper bound by the \emph{Shannon} entropy.
More precisely, we introduce a novel algorithm that, for every constant $c < \frac 1 2$, recovers $cm$ keys with amortized cost per key
$$\softO\left( 2^{\ent(\dist)} \right) = \softO\left( 2^{\ent(\chi)n} \right).$$
Furthermore, for $\frac 1 2 \leq c < 1$, the runtime increases only by a subexponential factor $2^{\mathcal{O}(\sqrt{n})}$ to $2^{\ent(\chi)n+\mathcal{O}(\sqrt{n})}$.
The key idea here is to abort key guessing after a certain number of guesses to avoid spending too much on time on unlikely keys.

    \begin{table}[htb]
        \centering
        \begin{tabular}{c|ccc}
             Scenario &  \ \ $0 < c \leq 1$ \ \ & Cost (Classic) & Cost (Quantum)  \\[0.2cm]
             \hline \\[-0.2cm]
             \multirow{2}{*}{constant fraction $c$} &  $c < \frac 1 2$   & \makecell{ $\softO(2^{\ent(\chi)n})$ \\ (\Cref{theo:const})} & 
             \makecell{$\softO(2^{\frac{\ent(\chi)n}{2}})$ \\ (\cref{theo:qConst}) }
             \\[0.5cm]
              & $\frac 1 2 \leq c < 1$ & \makecell{$2^{\ent(\chi)n + \softO(\sqrt{n})}$ \\ (\Cref{theo:const})} & \makecell{ $2^{\frac{\ent(\chi)n}{2} + \softO(\sqrt{n})}$ \\ (\cref{theo:qConst})}
              \\[0.5cm]
             \hline \\[-0.2cm]
             all-keys, & \multirow{2}{*}{$c=1$} & \multirow{2}{*}{\makecell{$\softO \left( 2^{\ent_{1/2}(\chi)n} \right)$ \\ (\cite{arikan1996inequality},\Cref{theo:arikan}) }} & \multirow{2}{*}{\makecell{ $\softO \left( 2^{\frac{\ent_{2/3}(\chi)n}{2}} \right)$ \\ (\Cref{theo:quantumArikan}) } } \\
             single key ($m=1$) \\[0.55cm]
        \end{tabular}
        \vspace*{0.5cm}
        \caption{Multi-key: guessing cost per key for $cm$ out of $m$ keys from $\D=\chi^n$.}
        \label{tab:cases_runtime}
    \end{table}

Our algorithm admits for a quantum version with runtimes $\softO\left( 2^{\frac{\ent(\chi)n}{2}} \right)$ and $2^{\frac{\ent(\chi)n}{2}+\mathcal{O}(\sqrt{n})}$, respectively.
Since for any non-uniform $\chi$ we have\linebreak$\ent(\chi) < \ent_{2/3}(\chi) < \ent_{1/2}(\chi)$, we thus substantially improve over the runtimes $2^{\ent_{1/2}(\chi)n}$ and $2^{\frac{\ent_{2/3}(\chi)n}{2}}$ in the single-key setting. Our results are summarized in \Cref{tab:cases_runtime}.

\paragraph{Practical Application.} Our multi-key guessing algorithm already found application in a recent lattice-based hybrid attack~\cite{ac25hybrid} on LWE-type schemes that combines lattice reduction with our multi-key guessing. \cite{ac25hybrid} shows that lattice hybrid attacks not only asymptotically outperform pure lattice reduction, but also lead to the currently best LWE attacks for many practical parameter settings~\cite{CHES:HRSS17,bos2018crystals}.

\medskip

\paragraph{Organization of the paper.} In \Cref{sec:prelim} we introduce Arikan's result. \Cref{sec:classical} studies classical key guessing, both in the single-key and in the multi-key setting. In \Cref{sec:quantum} we analyze Montanaro's algorithm for the single-key setting and provide our quantum algorithm for the multi-key setting. \Cref{sec:distribution} studies cryptographically relevant distributions and their quantum speed-ups for the single-key setting.

\section{Preliminaries}
\label{sec:prelim}

For a positive integer $n$, we set $[n] := \{1,\ldots,n\}$.
All logarithms are base-2.
For a cryptographic key~$k$ sampled from key space $\keys$, we define $\mathbbm{1}_k: \keys \to \{0,1\}$ with $\mathbbm{1}_k(x) = 1$ if and only if $x = k$.
Throughout the paper, all probability distributions are discrete.
We write $X \leftarrow \dist$ to denote that a random variable $X$ is drawn from some probability distribution~$\dist$.
Expected value and variance of $X$ are denoted by $\EE[X]$ and $\Var[X]$, respectively.
For a probability distribution $\dist$ with support $\keys$, its \emph{probability mass function} $P_\dist$ is defined as
$$
    P_\dist: \keys \to (0,1],\; k \mapsto \Pr_{X \leftarrow \dist}[ X = k ].
$$
\begin{definition}[Entropy]
    \label{def:entropy}
    Let $\dist$ be a probability distribution with support $\keys$.
    \begin{enumerate}
        \item For $\alpha > 0$, $\alpha \neq 1$, the \emph{Rényi entropy} of $\dist$ is defined as
        $$
            \ent_\alpha(\dist) := \frac{1}{1-\alpha}  \log\left( \sum_{x \in \keys} P_\dist(x)^\alpha \right).
        $$
        \item The \emph{Shannon entropy} of $\dist$ is defined as
        $$
            \ent(\dist) := \underset{X \leftarrow \dist}{\EE}[ -\log(P_\dist(X)) ] = - \sum_{x \in \keys} P_\dist(x) \log(P_\dist(x)).
        $$
        \item The \emph{min-entropy} of $\dist$ is defined as
        $$
            \ent_\textrm{\normalfont min}(\dist) := \min_{x \in \keys} -\log(P_\dist(x)).
        $$
    \end{enumerate}
\end{definition}
We extend the definition of $\ent_\alpha(\dist)$ for $\alpha\in \{0,1,\infty\}$ by setting
\[\renta\alpha\dist=\lim_{\alpha'\to\alpha}\renta{\alpha'}\dist\]
and find
\[\renta0\dist=\log\vert \keys\vert,\quad\renta1\dist=\ent(\dist), \quad \renta\infty\dist = \ent_\text{min}(\dist).\]
Note that for product distributions $\dist = \chi^n$, we have $\ent_{\alpha}(\dist) = n \ent_{\alpha}(\chi)$ for every~$\alpha$.

\begin{lemma}[Arikan's Inequality \cite{arikan1996inequality}]
    \label{lemma:Arikan}
    Let $\dist$ be a probability distribution.
    Let $p_1 \geq p_2 \geq \ldots$ denote the values assumed by the probability mass function $P_\dist$.
    Then for every $\rho > 0$, it holds that
    $$
        \frac{2^{\rho\renta{\frac{1}{1+\rho}}\dist}}{(1+\log \vert \keys\vert)^\rho} \leq \sum_{i} i^{\rho} \cdot p_i \leq 2^{\rho\renta{\frac{1}{1+\rho}}\dist}.
    $$
\end{lemma}
\begin{lemma}[Berry-Esseen Theorem \cite{berry1941accuracy,esseen1945fourier}]
    \label{lemma:BerryEsseen}
    Let $X_1,X_2,\ldots$ be a sequence of i.i.d. random variables with $\EE[X_i] < \infty$, $0 < \Var[X_i] < \infty$ and $\EE[|X_i|^3] < \infty$.
    Define $\mu := \EE[X_i]$, $\sigma^2 := \Var[X_i]$ and $\overline{X_n} := \frac{1}{n} \sum_{i=1}^n X_i$.
    Then the distribution of $\sqrt{n}( \overline{X_n} - \mu )$ converges to a Gaussian distribution with mean~0 and variance $\sigma^2$ at rate $\mathcal{O}(1/\sqrt{n})$.
    That is, for every $t \in \RR$ it holds that
    $$
        \Pr[ \sqrt{n}( \overline{X_n} - \mu ) \leq t ] = \int_{-\infty}^t \frac{1}{\sigma \sqrt{2\pi}} \operatorname{exp}\left({ - \frac{x^2}{2\sigma^2}}\right) \textrm{\textup d}x \pm \mathcal{O}\left( \frac{1}{\sqrt{n}} \right).
    $$
\end{lemma}
\begin{lemma}[Grover's Algorithm \cite{Grover96,hoyer2000arbitrary,brassard2002quantum}]
    \label{lemma:Grover}
    Let $\ket{\Psi}$ be a uniform superposition over some finite set $\keys$, and let $\tau : \keys \to \{0,1\}$ be a function, such that $\tau(x) = 1$ for at most one $x \in \keys$.
    Given $\ket{\Psi}$ and oracle access to $\tau(\cdot)$, Grover's algorithm outputs $x \in \keys$ with $\tau(x) = 1$, if it exists, and an error symbol $\bot$ otherwise. Grover's algorithm achieves this, using $\lceil \frac \pi 4 \sqrt{|\keys|} \rceil + 1$ queries to $\tau$.
\end{lemma}
\section{Classical Key Guessing}
\label{sec:classical}


\subsection{Single-Key Guessing}
\label{sec:singleKey}

\begin{definition}[Single-Key Guessing Problem]
    \label{def:singleKey}
    Let $\dist$ be a probability distribution with support $\keys$.
    Let $k \leftarrow \dist$.
    The \emph{single-key guessing problem} is defined as follows. Given a description of $\dist$ and oracle access to $\oracle(\cdot)$, the goal is to find the {key} $k$.
\end{definition}
The optimal strategy for solving the single-key guessing problem with success probability $1$ is to enumerate all possible keys in decreasing order of probability, until the correct key is found. This strategy requires access to an efficient algorithm $\keyenum$, that on input $i \in \{1,\ldots,|\keys|\}$ outputs the $i$-th most likely key.
For the common cryptographic setting, where $\dist$ is a product distribution $\dist = \chi^n$ over some finite set $\keys = A^n$, Budroni and M{\aa}rtenson~\cite{budMar23} give an efficient instantiation of such an algorithm, that (after one initial pre-computation phase running in time $\widetilde{\mathcal{O}}(n^{|A|-1})$) outputs the $i$-th key in time $\mathcal{O}(|A| \cdot n)$. 
Throughout the paper, we assume black box access to $\keyenum$.

\begin{algorithm}[htb]
    \KwIn{Key guessing instance $(\dist,\oracle)$, access to algorithm $\keyenum$}
    \KwOut{ Key $k$ }
    

    $i \leftarrow 0$
    
    
    \Repeat{$\oracle(x) = 1$}{
      
      $i \leftarrow i+1$
      
      $x \leftarrow \keyenum(i)$
      
    }
    
    \Return{$x$}

    \caption{\fullguess{ } }
    \label{alg:FullEnumSimple}
\end{algorithm}

The optimal algorithm for the single-key guessing problem is given in \cref{alg:FullEnumSimple}.
As shown by Arikan~\cite{arikan1996inequality}, the Rényi entropy tightly bounds the runtime $T_\mathsf{KG}$ of \fullguess{} (\cref{alg:FullEnumSimple}).
\begin{theorem}[\cite{arikan1996inequality}]
    \label{theo:arikan}
    Let $\dist$ be a distribution with support $\keys$.
    On input of a single-key guessing instance $(\dist,\oracle)$, {\fullguess{}} outputs the correct key $k$ in expected time
    $$
         \frac{2^{\ent_{ 1 / 2}(\dist)}}{1+\log|\keys|}  \leq \EE[T_\mathsf{KG}] \leq 2^{\ent_{ 1 / 2}(\dist)}.
    $$
\end{theorem}
\begin{proof}
    Let $p_i$ denote the probability of the $i$-th likeliest key drawn from $\dist$.
    Then $\Pr[T_\mathsf{KG} = i] = p_i$, and thus $\EE[T_\mathsf{KG}] = \sum_i i \cdot p_i$.
    By applying Arikan's inequality (\cref{lemma:Arikan}), the desired statement follows.\qed
\end{proof}

\subsection{Multi-Key Guessing}
\label{sec:multiKey}

The \emph{multi-key guessing problem} asks to find a fraction $cm$ out of $m$ keys.
More precisely, the problem is defined as follows.
\begin{definition}[Multi-Key Guessing Problem]
    Let $\dist$ be a probability distribution with support $\keys$.
    Let $(k_1,\ldots,k_m) \leftarrow \dist^m$.
    The \emph{multi-key guessing problem} with parameter $0 < c < 1$ is defined as follows.
    Given a description of $\dist$ and oracle access to $\mathbbm{1}_{k_1}(\cdot),\ldots,\mathbbm{1}_{k_m}(\cdot)$, the goal is to find a tuple\linebreak$(k_1',\ldots,k_m') \in (\keys \cup \{\bot\})^m$ such that $k_i' = k_i$ for at least $cm$ keys $k_i$.
\end{definition}
Before we can give our algorithm for solving the \emph{multi-key} guessing problem, we have to introduce with \abortguess{} (\cref{alg:aborted}) another algorithm for the \emph{single-key} guessing problem.
As its name suggests, \abortguess{} is a variant of \fullguess{} that aborts if it does not find the key $k$ within the first $t$ trials, where $t$ is some runtime parameter.
\begin{algorithm}[htb]
    \KwIn{Key guessing instance $(\dist,\oracle)$, runtime parameter $t$, access to algorithm $\keyenum$}
    \KwOut{Key $k$ or error symbol $\bot$.}

    $i \leftarrow 0$

    \Repeat{$\oracle(x) = 1$ \normalfont or $j > t$}{
      
      $i \leftarrow i+1$
      
      $x \leftarrow \keyenum(i)$
      
    }

    \lIf{$\mathbbm{1}_{k_j}(x) = 1$}{
        \Return{$x$}
    }
    \lElse{
        \Return{$\bot$}
    }

    \caption{\abortguess{}}
    \label{alg:aborted}
\end{algorithm}

\begin{algorithm}[htb]
    \KwIn{Multi-key guessing instance $(\dist,\multiOracle)$, parameter $c$, access to algorithm $\keyenum$}
    \KwOut{$(k_1',\ldots,k_m') \in (\keys \cup \{\bot\})^m$ with $k_i' = k_i$ for at least $cm$ keys $k_i$.}

    $(k_1',\ldots,k_m') \leftarrow (\bot,\ldots,\bot)$
    
    $t \leftarrow 1$

    \Repeat{\normalfont $k_j' \neq \bot$ for at least $cm$ keys $k_j'$}{

        $t \leftarrow 2t$
    
        \For{$j \in [m]$}{
    
            $k_j' \leftarrow \abortguess((\dist,\mathbbm{1}_{k_j}),t)$
            
        }
    }
    
    \Return{$(k_1',\ldots,k_m')$}

    \caption{\multiguess{}}
    \label{alg:multiGuess}
\end{algorithm}

Using \abortguess{}, we can now introduce our algorithm \multiguess{} (\cref{alg:multiGuess}) for solving the multi-key guessing problem.
Given a multi-key guessing instance $(\dist,\multiOracle)$, \multiguess{} runs \abortguess{} repeatedly on each single-key instance $(\dist,\mathbbm{1}_{k_j})$ with exponentially increasing runtime parameter $t = 2,4,8,\ldots$ until it recovers at least $cm$ keys $k_j$.

Our first major result is the following theorem, which shows that for product distributions $\dist = \chi^n$ the \emph{Shannon} entropy bounds the runtime of \multiguess{}.
Compared to the single-key setting (\cref{theo:arikan}), the runtime thus improves by a factor at least
$$
    \frac{ 2^{\ent(\chi)n} }{ 2^{ \ent_{1/2}(\chi)n } }.
$$
For every non-uniform $\chi$, this is exponential in $n$.
\begin{restatable}{theorem}{theoMultiGuess}
    \label{theo:const}
    Let $\dist$ be a product distribution $\dist = \chi^n$, where $\chi$ has constant support.
    Let $(\dist,\multiOracle)$ be a multi-key guessing instance with constant parameter $0 < c < 1$.
    \begin{enumerate}
        \item If $c < \frac{1}{2}$, then with probability at least $1 - e^{-\Omega(m)}$, {\multiguess{}} solves $(\dist,\multiOracle)$ with amortized cost per key
        \begin{equation*}
            T_\mathsf{MKG} = \softO\left(2^{\ent(\chi)n}\right).
        \end{equation*}
        \item For $\frac{1}{2} \leq c < 1$, the amortized cost increases at most by a subexponential factor $2^{\mathcal{O}(\sqrt{n})}$ to $2^{\ent(\chi)n+\mathcal{O}(\sqrt{n})}$.
    \end{enumerate}
\end{restatable}

\subsection{Proof for \cref{theo:const}}
To be able to prove \cref{theo:const}, we have to introduce the following novel definition.
\begin{definition}[Core Set]
    Let $\chi$ be a probability distribution with support $A$.
    The \emph{core set} $\coreset$ of $\chi$ with parameters $n \in \NN$ and $\delta \geq 0$ is defined as
    $$
        \coreset := \left\{ a \in A^n \mid P_{\chi^n}(a) \geq 2^{-\ent(\chi)n-\delta n} \right\}.
    $$
\end{definition}
Note that our core set contains the \emph{typical set}
$$
    \typical := \left\{ a \in A^n \mid 2^{-\ent(\chi)n+\delta n } \geq P_{\chi^n}(a) \geq 2^{-\ent(\chi)n-\delta n} \right\} \subseteq \coreset,
$$
which is often used in information theory, see~\cite[Chapter 3]{cover2006elements}.
An important property of the typical set is that, for every \emph{constant} $\delta > 0$ and $n \to \infty$, the probability mass $P_{\chi^n}( \typical )$ of the typical set converges to 1.
Our core set allows for the following more fine-grained statement.
\begin{lemma}
    \label{lemma:typical}
    Let $\chi$ be a probability distribution with finite support.
    For every $\delta \geq 0$, it holds that
    \begin{equation}
        \label{eq:typicalUpper}
        \Pr_{k \leftarrow \chi^n}\left[ k \in \coreset \right] \geq \frac{1}{2} \pm \mathcal{O}\left( n^{-\frac{1}{2}} \right).
    \end{equation}
    Furthermore, for every $\varepsilon \in (0,1)$, there exists $\delta = \Theta\left( \sqrt{\frac{\ln(\varepsilon^{-1})}{n}} \right)$ such that
    \begin{equation}
        \label{eq:typicalProb}
        \Pr_{k \leftarrow \chi^n}\left[ k \in \coreset \right] \geq 1 - \varepsilon.
    \end{equation}

\end{lemma}
\begin{proof}
    Let $k = (k_1,\ldots,k_n) \leftarrow \chi^n$, and let $A$ denote the support of $\chi$.
    If $\chi$ is the uniform distribution on $A$, then $\ent(\chi) = \log(|A|)$, and therefore
    $$
        P_{\chi^n}(a) = |A|^{-n} = 2^{-{\ent(\chi)n}},
    $$
    for every $a \in A^n$.
    Hence, for uniform $\chi$, every $a \in A^n$ lies in $\coreset$ for arbitrary $\delta\ge0$, and we have $\Pr\left[ k \in \typical \right] = 1$.
    In particular, both \cref{eq:typicalUpper,eq:typicalProb} hold.
    
    It remains to prove \cref{eq:typicalUpper,eq:typicalProb} for non-uniform $\chi$.
    By definition of $\coreset$ it holds that
    \begin{equation*}
        \Pr\left[ k \in \coreset \right] = \Pr\left[ \prod_{i=1}^n P(k_i) \geq 2^{-\ent(\chi)n - \delta n} \right].
    \end{equation*}
    Let $X_i := -\log P(X_i)$.
    We set $\overline{X_n} := \frac{1}{n} \sum_{i=1}^n X_i$, and rewrite the above probability as
    \begin{equation*}
        \Pr\left[ k \in \coreset \right]  = \Pr\left[ - \sum_{i=1}^n X_i \geq -{\ent(\chi)n} - \delta n \right] = \Pr\left[ \overline{X_n} - {\ent(\chi)} \leq \delta \right].
    \end{equation*}
    We now make three important observations:
    \begin{enumerate}
        \item\label{cond1} By definition of Shannon entropy, $\mathbb{E}[X_i] = \ent(\chi) < \infty$.
        \item\label{cond2} Since $\chi$ is not uniform, $X_i$ is not constant and thus $\Var[X_i] > 0$. 
        \item\label{cond3} Since $\chi$ has finite support, both $\Var[X_i]$ and $\EE[|X_i|^3]$ are finite.
    \end{enumerate}
    By the Berry-Esseen Theorem (\cref{lemma:BerryEsseen}), the distribution of $\sqrt{n}(\overline{X_n} - {\ent(\chi)})$ thus converges at rate $\mathcal{O}(1/\sqrt{n})$ to a Gaussian distribution with mean 0 and variance $\sigma^2 := \Var[X_i]$.
    Hence,
    \begin{align*}
        \Pr\left[ k \in \coreset \right] &= \Pr\left[ \overline{X_n} - {\ent(\chi)} \leq \delta \right] = \Pr\left[  \sqrt{n}(\overline{X_n} - {\ent(\chi)}) \leq \delta \sqrt{n} \right]\\
        &= \int_{-\infty}^{\delta \sqrt{n}} \frac{1}{\sigma \sqrt{2\pi}} \operatorname{exp}\left({ - \frac{x^2}{2\sigma^2}}\right) \textrm{d}x  \pm \mathcal{O}\left( \frac{1}{\sqrt{n}} \right)\\
        &\geq \int_{-\infty}^0 \frac{1}{\sigma \sqrt{2\pi}} \operatorname{exp}\left({ - \frac{x^2}{2\sigma^2}}\right) \textrm{d}x  \pm \mathcal{O}\left( \frac{1}{\sqrt{n}} \right)\\
        &= \frac 1 2  \pm \mathcal{O}\left( \frac{1}{\sqrt{n}} \right),
    \end{align*}
    which proves \cref{eq:typicalUpper}.
    
    In order to prove \cref{eq:typicalProb}, simply apply Hoeffding's inequality to conclude that
    $$
        \Pr\left[k \in \coreset \right] = \Pr\left[ \overline{X_n} - {\ent(\chi)} \leq \delta \right]  \geq 1 - \exp\left( -  \Theta\left(\delta^2 n\right) \right).
    $$
    \qed
\end{proof}
\begin{lemma}
    \label{lemma:sizeCore}
    Let $\chi$ be a probability distribution.
    For every $\delta \geq 0$, it holds that
    $$
        |\coreset| \leq 2^{\ent(\chi)n + \delta n}.
    $$
\end{lemma}
\begin{proof}
    Our proof is analogous to the proof of \cite[Theorem 3.1.2]{cover2006elements}, which gives the same upper bound on the size of the typical set $\typical$.
    By definition of $\coreset$, we have
    \begin{align*}
        1 &= \sum_{a \in A^n} P_{\chi^n}(a) \geq \sum_{a \in \coreset} P_{\chi^n}(a) \geq \sum_{a \in \coreset} 2^{- \ent(\chi)n - \delta n }\\
        &= |\coreset| \cdot 2^{- \ent(\chi)n-\delta n}.
    \end{align*}
    Multiplying the above inequality by $2^{\ent(\chi)n+\delta n}$, the lemma follows.\qed
\end{proof}
Intuitively, \cref{lemma:typical,lemma:sizeCore} together show that the $2^{\ent(\chi)n}$ most likely keys make up a large fraction of the probability mass of $\chi^n$.
Using this observation, we are now ready to prove \cref{theo:const}, which we restate below for better readability.
\newpage
\theoMultiGuess*
\begin{proof}
    For $\alpha \in \NN$, let
    $$
        \mathcal{M}(\alpha) := \left\{ \keyenum(i) \mid i \in \left\{1,2,\ldots, 2^{\alpha} \right\} \right\}
    $$
    denote the set of the $2^{\alpha}$ most likely keys.
    \multiguess{} queries the oracles $\multiOracle$ on all elements of $\mathcal{M}(1),\mathcal{M}(2),\ldots$ until it reaches some $\mathcal{M}(\alpha_\text{max})$ that contains at least $cm$ keys $k_i$.
    The runtime of \multiguess{} is thus
    $$
        \sum_{\alpha=1}^{\alpha_\text{max}} m \cdot 2^\alpha \leq \alpha_\text{max} \cdot m \cdot 2^{\alpha_\text{max}} = \softO\left( m \cdot 2^{\alpha_\text{max}} \right).
    $$
    In other words, the amortized cost per key is $T_\mathsf{MKG} = \softO\left( 2^{\alpha_\text{max}} \right)$.
    Hence, to prove the theorem, it suffices to prove the following statements:
    \begin{enumerate}
        \item If $c < \frac{1}{2}$ and $n$ is sufficiently large, then
        \begin{equation}
            \label{eq:smallC}
            \Pr\left[ \alpha_\text{max} \leq \left\lceil\ent(\chi)n\right\rceil \right] \geq 1 - e^{-\Omega(m)}.
        \end{equation}
        \item If $\frac{1}{2} \leq c < 1$, then
        \begin{equation}
            \label{eq:largeC}
            \Pr\left[ \alpha_\text{max} \leq \ent(\dist)n + \Theta(\sqrt{n}) \right] \geq 1 - e^{-\Omega(m)}.
        \end{equation}
    \end{enumerate}
    For $\delta \geq 0$ and $j \in [m]$, let $X_{j,\delta} \in \{0,1\}$ denote an indicator variable with
    $$X_{j,\delta} = 1 \iff k_j \in \mathcal{M}(\left\lceil \ent(\chi)n+\delta n \right\rceil).$$
    Then
    \begin{equation}
        \label{eq:maxAsExp}
        \Pr\left[ \alpha_\text{max} \leq \left\lceil \ent(\chi)n+\delta n \right\rceil \right] = \Pr\left[ \sum_{j=1}^m X_{j,\delta} \geq cm \right].
    \end{equation}
    From the definition of $\coreset$ and \cref{lemma:sizeCore} it follows that $\coreset \subseteq \mathcal{M}(\left\lceil \ent(\chi)n+\delta n \right\rceil)$.
    Hence, for $\mu_\delta := \EE[X_{j,\delta}]$, we have
    \begin{equation}
        \label{eq:coreContained}
        \mu_\delta = \Pr[X_{j,\delta}= 1] = \Pr\left[k_j \in \mathcal{M}\left( \alpha_\delta \right) \right] \geq \Pr\left[k_i \in \coreset \right].
    \end{equation}
    We now prove the result for $c < \frac{1}{2}$, i.e., \cref{eq:smallC}.
    If $c < \frac{1}{2}$, then there exists some constant $c' \in (c,\frac{1}{2})$.
    By \cref{eq:coreContained} and \cref{lemma:typical} (specifically \cref{eq:typicalUpper}) we have, for sufficiently large $n$, that
    $$
        \mu_0 \geq \Pr\left[k_i \in \mathcal{C}_{\chi}^{n,0} \right] \geq {c'}.
    $$
    Let $\varepsilon := 1 - \frac{c}{c'} \in (0,1)$.
    Then $cm = (1-\varepsilon)dm \leq (1-\varepsilon)\mu_0 m$.
    Hence, by \cref{eq:maxAsExp} and a Chernoff bound, we have
    $$
        \Pr\left[ \alpha_\text{max} \leq \left\lceil \ent(\chi)n \right\rceil \right] \geq \Pr\left[ \sum_{j=1}^m X_{j,\delta} \geq (1-\varepsilon)\mu_0m \right] \geq 1 - e^{-\Omega(m)},
    $$
    which proves \cref{eq:smallC}.

    It remains to prove the result for $c \in [\frac{1}{2},1)$, i.e.,  \cref{eq:largeC}.
    If $c \in [\frac{1}{2},1)$, then there exists some constant $c' \in (c,1)$.
    By \cref{eq:coreContained} and \cref{lemma:typical} (specifically \cref{eq:typicalProb}), there exists $\delta = \Theta\left( \sqrt{n} \right)$ such that
    $$
        \mu_\delta \geq c'.
    $$
    Now using the same argument as above, we obtain
    $$
        \Pr\left[ \alpha_\text{max} \leq  \ent(\chi)n + \Theta\left(\sqrt{n}\right) \right] \geq 1 - e^{-\Omega(m)},
    $$
    which proves \cref{eq:largeC}.
    \qed
\end{proof}
\section{Quantum Key Guessing}
\label{sec:quantum}


\subsection{A Tight Analysis for Montanaro's Single-Key Algorithm}
Montanaro~\cite{montanaro2011quantum} proved the following result for the quantum complexity of the single-key guessing problem.
\begin{theorem}[Propositions 2.1 and 2.4 in \cite{montanaro2011quantum}]
    \label{theo:montanaro}
    There is a quantum algorithm that solves single-key guessing instances $(\dist,\oracle)$ in expected time at most
    $$
        e \pi \sum_{i} \sqrt{i} p_i,
    $$
    where $p_1 \geq p_2 \geq \ldots$ are the values assumed by the probability mass function~$P_\dist$.
    Up to constant factors, this is the optimal runtime for solving the single-key guessing problem quantumly.
\end{theorem}
By Cauchy-Schwarz, the quantum complexity $T_\mathsf{QKG} := \sum_i \sqrt{i}p_i$ of single-key guessing is upper bound by
$$
    T_\mathsf{QKG} = \sum_i \sqrt{i} p_i = \sum_i \sqrt{i} \sqrt{p_i} \sqrt{p_i} \leq \sqrt{ \sum_i ip_i } \sqrt{ \sum_i p_i } =  \sqrt{ \sum_i ip_i }.
$$
Hence, for any distribution $\dist$, Montanaro's algorithm achieves \emph{at least} a quadratic speed-up over the optimal classical algorithm \fullguess{}, which has runtime $\sum_i i p_i$.
However, it seemed unclear so far, whether the algorithm achieves a \emph{super-quadratic} square-root for cryptographically relevant distributions.
As the following novel theorem shows, this is indeed the case.
\begin{theorem}[Montanaro Runtime]
    \label{theo:quantumArikan}
    Let $\dist$ be a distribution with support $\keys$.
    On input of a single-key guessing instance $(\dist,\oracle)$, Montanaro's algorithm outputs the correct key $k$ in expected time
    $$
        \frac{2^{\frac{\ent_{2/3}(\dist)}{2}}}{\sqrt{1+\log|\keys|}} \leq \EE[T_\mathsf{QKG}] \leq 2^{\frac{\ent_{2/3}(\dist)}{2}}.
    $$
\end{theorem}
\begin{proof}
    Applying Arikan's inequality (\cref{lemma:Arikan}) with $\rho = \frac{1}{2}$ to \cref{theo:montanaro} immediately gives the bounds  on $\EE[T_\mathsf{QKG}]$.
    \qed
\end{proof}

\begin{theorem}[Quantum Speed-up]
    \label{theo:speedup}
    Let $\dist$ be a distribution $\dist$ with support $\keys$.
    On input of a single-key guessing instance $(\dist,\oracle)$,
    Montanaro's algorithm achieves over the optimal classical algorithm quantum speed-up
    $$
        s \geq 2 \cdot \frac{\ent_{1/2}(\dist) - \log(1+\log|\keys|)}{\ent_{2/3}(\dist)}.
    $$
    For a product distribution $\D=\chi^n$ with finite support we obtain
    \[
       s \geq 2 \cdot \frac{\ent_{1/2}(\chi)}{\ent_{2/3}(\chi)}\left(1-\mathcal{O}\left(\frac{\log n}{n}\right)\right).
    \]            
\end{theorem}
\begin{proof}
    Let $\TC$ and $\TQ$ be the expected runtimes of the optimal classic \fullguess{} algorithm and Montanaro's quantum algorithm, respectively. Using \Cref{theo:arikan,theo:quantumArikan}, we obtain quantum speed-up
    \[
    s = \frac{\log(\TC)}{\log(\TQ)} \geq 2 \cdot \frac{\ent_{1/2}(\dist)-\log(1+\log|\K|)}{\ent_{2/3}(\D)}.
    \]
    For a product distribution $\D = \chi^n$ with finite support we may rewrite the speed-up as
    \[
    s \geq 2 \cdot \frac{n\ent_{1/2}(\chi)-\log(\mathcal{O}(n))}{n\ent_{2/3}(\chi)} = 2 \cdot \frac{\ent_{1/2}(\chi)}{\ent_{2/3}(\chi)}\left(1-\mathcal{O}\left(\frac{\log n}{n}\right)\right).
    \]
    \qed
\end{proof}

Recall that for all $\chi$ different from the uniform distribution, we have\linebreak $\ent_{1/2}(\chi) > \ent_{2/3}(\chi)$.
Hence, \cref{theo:speedup} shows that Montanaro's algorithm asymptotically achieves a super-quadratic speed-up for any non-uniform $\chi$.

\subsection{Quantum Multi-Key Guessing}
In \cref{sec:classical}, we showed that the amortized cost for guessing many keys is $\softO\left( 2^{\ent(\chi)n} \right)$, and thus significantly below the cost $\softO\left( 2^{\ent_{\frac 1 2}(\chi)n} \right)$ of guessing a single key.
In this section, we show a similar speed-up for the quantum setting.
To this end, we introduce the algorithm \qguess{} (\cref{alg:qGuess}), which is the quantum analogue of \multiguess{} (\cref{alg:multiGuess}).

\begin{algorithm}[htb]
    \KwIn{Multi-key guessing instance $(\dist,\multiOracle)$, parameter $c$, access to algorithm $\keyenum$}
    \KwOut{$(k_1',\ldots,k_m') \in (\keys \cup \{\bot\})^m$ with $k_i' = k_i$ for at least $cm$ keys $k_i$.}

    $(k_1',\ldots,k_m') \leftarrow (\bot,\ldots,\bot)$
    
    $t \leftarrow 1$

    \Repeat{\normalfont $k_j' \neq \bot$ for at least $cm$ keys $k_j'$}{

        $t \leftarrow 2t$
    
        \For{$j \in [m]$}{

            Initialize superposition $\ket{\Psi} \leftarrow 1/\sqrt{t} \cdot \sum_{i=1}^{t} \ket{i}$.
            
            Run Grover's algorithm on $\ket{\Psi}$ with oracle $\mathbbm{1}_{k_j}(\keyenum(\cdot))$.
            
            If Grover's algorithm did not return $\bot$, apply $\keyenum{}$ to the result and set $k_j'$ to the resulting key.
            
        }
    }
    
    \Return{$(k_1',\ldots,k_m')$}

    \caption{\qguess{}}
    \label{alg:qGuess}
\end{algorithm}

Recall that our classical algorithm \multiguess{} uses\linebreak \abortguess{} to query the oracles $\mathbbm{1}_{k_j}$ on the $t$ most likely keys, for some exponentially increasing parameter $t$.
Analogously, our quantum algorithm runs Grover's algorithm on the $t$ most likely keys.
As a a result, we obtain a square-root speed-up over \multiguess{}'s runtime from \cref{theo:const}.
In particular, we have the following theorem.
\begin{theorem}
    \label{theo:qConst}
    Let $\dist$ be a product distribution $\dist = \chi^n$, where $\chi$ has constant support.
    Let $(\dist,\multiOracle)$ be a multi-key guessing instance with constant parameter $0 < c < 1$.
    \begin{enumerate}
        \item If $c < \frac{1}{2}$, then with probability at least $1 - e^{-\Omega(m)}$, {\qguess{}{}} solves $(\dist,\multiOracle)$ with amortized cost per key
        \begin{equation*}
            T_\mathsf{QMKG} = \softO\left(2^{\frac{\ent(\chi)n}{2}}\right).
        \end{equation*}
        \item For $\frac{1}{2} \leq c < 1$, the amortized cost increases at most by a subexponential factor $2^{\mathcal{O}(\sqrt{n})}$ to $2^{\frac{\ent(\chi)n}{2}+\mathcal{O}(\sqrt{n})}$.
    \end{enumerate}
\end{theorem}

\section{Quantum Speed-ups for Various Distributions }
\label{sec:distribution}

Using \Cref{theo:speedup}, we compute the quantum speed-up $s$ of Montanaro's quantum algorithm over classical key guessing for several distributions of interest.

\paragraph{Bernoulli.} The Bernoulli distribution $\chi=\text{Ber}(p)$ with $X \sim \text{Ber}(p)$ satisfies 
\[
  \Pr[X=1]=p \textrm{ and } \Pr[X=0]=1-p. 
\]
Consider keys sampled from $\D=\chi^n$ as in LPN~\cite{STOC:BluKalWas00}.  The results of applying \Cref{theo:speedup} are depicted in \Cref{fig:bernoulli}.

\begin{wrapfigure}{r}{0.45\textwidth}
     \centering
    \includegraphics[width=0.45\textwidth]{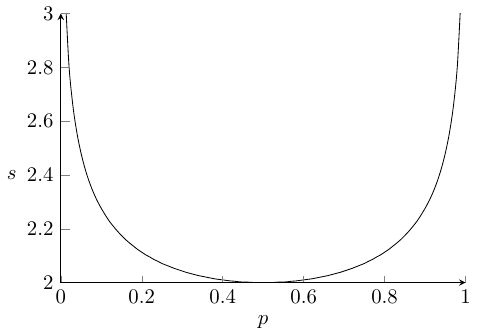}
    \vspace*{-0.8cm}
    \caption{Bernoulli $\text{Ber}(p)$ Speed-up }
    \label{fig:bernoulli}
\end{wrapfigure}

Notice that $\text{Ber}(\frac 12)$ is the uniform distribution over $\{0,1\}^n$. For any \linebreak $p \ne \frac 1 2$, we obtain super-quadratic quantum speed-ups $s$. As a numerical example, we have $s > 2.27$ for $p=0.1$. Notice that LPN typically uses small constant $p$~\cite{AC:HopBlu01}. LPN-based public key encryption~\cite{FOCS:Alekhnovich03}, requires an even q smaller choice of $p = \Theta(n^{-\frac 1 2})$. In this small noise LPN regime, we obtain a {\em super-polynomial} quantum speed-up $s=\Omega(n^{\frac 1 {12}})$.

\vspace{1cm}
\begin{lemma}[Small Noise LPN Quantum Speed-Up] On input of a single-key guessing instance $(\textrm{Ber}(p)^n,\oracle)$ with $p\le\frac12$, Montanaro's algorithm asymptotically achieves a quantum speed-up over classical key guessing $\fullguess{}$  satisfying
\[
  s \geq \frac 1 3 p^{- \frac 1 6}.
\]
Thus, for $p=\Theta(n^{-\frac 1 2})$ Montanaro's algorithm achieves speed-up $s = \Omega(n^{\frac 1 {12}})$.
\end{lemma}
\begin{proof}
According to \Cref{theo:arikan}, the runtime of $\fullguess{}$ is tightly bounded by $2^{H_{1/2}(\text{Ber}(p))n}$. For $p \leq \frac 1 2$ we obtain
\[
  H_{\frac 1 2}(\text{Ber}(p))n = 2 \log\left(p^{\frac 12} + (1-p)^{\frac 12}\right)n \geq 2 \log\left(1+\frac{p^{\frac 1 2}}{2}\right)n \geq \frac 1 2 \log(e)p^{\frac 1 2}n.
\]
By \Cref{theo:quantumArikan} the runtime of Montanaro's algorithm is upper bounded by $2^{\frac 1 2 H_{2/3}(\text{Ber}(p))n}$, where
\[
  H_{2/3}(\text{Ber}(p))n = 3  \log\left(p^{\frac 2 3} + (1-p)^{\frac 2 3} \right)n \leq 3  \log\left( 1 + p^{\frac 2 3} \right) n \leq 3  \log(e) p^{\frac 2 3} n.
\]
Therefore, we achieve quantum speed-up
\[
s = 2 \cdot \frac{H_{\frac 1 2}(\text{Ber}(p))n}{ H_{\frac 2 3}(\text{Ber}(p))n} \geq \frac 1 3  p^{-\frac 1 6}. \hfill 
\]
\qed
\end{proof}

\paragraph{Ternary.} 
We define the  Ternary distribution $\chi=\T(p)$,  where $X \sim \T(p)$ satisfies 
\[
  \Pr[X=(-1)]= \Pr[X=1] = \frac p 2 \textrm{ and } \Pr[X=0]=1-p.
\]
\begin{wrapfigure}{r}{0.45\textwidth}
     \centering
    \includegraphics[width=0.45\textwidth]{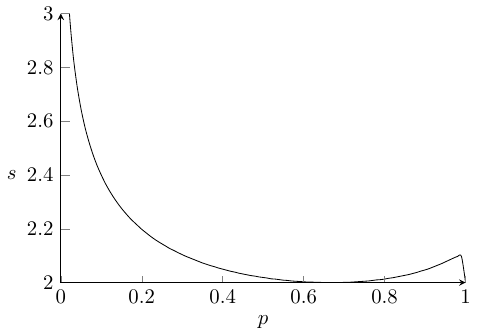}
    \vspace*{-0.8cm}
    \caption{Ternary $\mathcal{T}(p)$ Speed-up}
    \label{fig:ternary}  
\end{wrapfigure}
Consider keys sampled from\linebreak $\D=\chi^n$ as in NTRU-type schemes~\cite{HofPipSil98,CHES:HRSS17}. The results of applying our Quantum Speed-up theorem (\Cref{theo:speedup}) are provided in \Cref{fig:ternary}.

Notice that for $p=\frac 2 3$ we obtain the uniform distribution over $\{-1,0,1\}^n$, and for $p=1$ the uniform distribution over $\{-1,1\}^n$. For any $p \not\in \{\frac 2 3, 1\}$, we obtain super-quadratic quantum speed-ups $s$. As a numerical example, we have $s \approx 2.4$ for $p=0.1$. A typical NTRU choice is $\mathcal{T}(3/8)$ with $s \approx 2.06$. Again, the speed-up $s$ grows to infinity $s$ for $p \to 0$.

\medskip
\medskip

\paragraph{Discrete Gaussian.} We define a discrete Gaussian distribution $\chi=\mathcal{D_{S, \sigma}}$ with support $S$ and standard deviation $\sigma$, where $X \sim \mathcal{D_{S, \sigma}}$ satisfies for all $i \in S$
\[\Pr[X=i] = \frac{1}{c(S)} \cdot e^{\frac{-i^2}{2\sigma^2}} \textrm{ with }
c(S)=\sum_{j \in S} e^{\frac{-j^2}{2\sigma^2}} .
\]
\begin{wrapfigure}{r}{0.49\textwidth}
     \centering
    \includegraphics[width=0.49\textwidth]{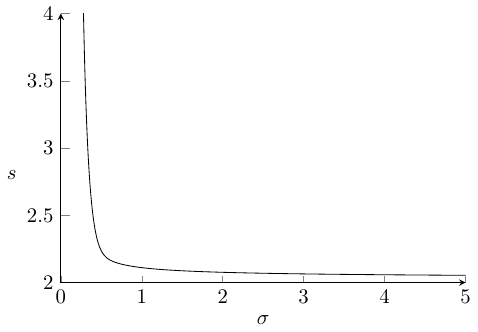}
    \vspace*{-0.8cm}
    \caption{Gauss $\mathcal{D}_{\{-100, \ldots, 100\}, \sigma}$ Speed-up}
    \label{fig:gauss}
\end{wrapfigure}

Consider keys sampled from\linebreak $\D=(\mathcal{D_{S, \sigma}})^n$ as in typical LWE schemes~\cite{STOC:Regev03}. The choice\linebreak $S = \mathbb{Z}$ results in infinite support for the discrete Gaussian. In order to apply our Quantum Speed-up theorem (\Cref{theo:speedup}) we choose a discrete Gaussian\linebreak$\chi=\mathcal{D}_{\{-100, \ldots, 100\}, \sigma}$ with finite support. 
The results are provided in \Cref{fig:gauss}.

For $\sigma=1$, we obtain $s > 2.11$. For large standard deviation $\sigma$ the quantum speed-up converges to $2$, since the discrete Gaussian approaches the uniform distribution. For small $\sigma$, the quantum speed-up $s$ becomes arbitrary large.

\paragraph{Binomial.} The Binomial distribution $\text{Bin}(m, \frac 1 2)$ with   $X \sim \text{Bin}(m,\frac 1 2)$ satisfies
\[
  \Pr[X=i] = \frac{\binom mi}{2^n} \textrm{ for } i=0, \ldots, m.
\]
\newpage
\begin{wrapfigure}{r}{0.45\textwidth}
     \centering
    \includegraphics[width=0.45\textwidth]{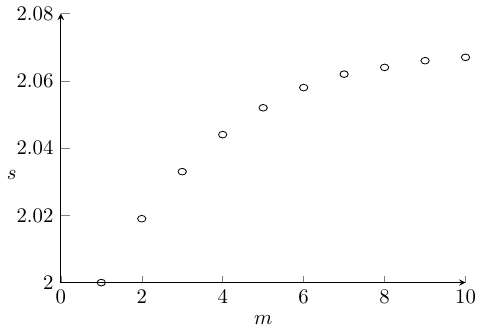}
    \vspace*{-0.8cm}
    \caption{Binomial $\mathcal{B}(m,\frac 12)$ Speed-Up}
    \label{fig:binomial}
\end{wrapfigure}
The Binomial distribution, usually centered around $0$, is used as an efficient replacement of a discrete Gaussian distribution in LWE-based schemes.  The distribution is easy to sample and leads to compact keys with a small support. As an example, Kyber~\cite{bos2018crystals} uses $\text{Bin}(4,\frac 1 2)$ and $\text{Bin}(6, \frac 1 2)$. 
The results of applying \Cref{theo:speedup} to $\text{Bin}(m, \frac 1 2)$ are shown in \Cref{fig:binomial}. For the Kyber choices $m=4$ and~$6$ Montanaro's algorithm achieves quantum speed-ups of 
$s > 2.04$ respectively  $s>2.05$. 

For all $m>1$, one obtains superquadratic speed-ups. For arbitrary large $m$, we still have $s \approx 2.04$.

\paragraph{Zipf.} 
The Zipf distribution $\chi=\Z(N,t)$ with $X \sim \Z(N,t)$ satisfies 
for all\linebreak $i=1, \ldots, N$
\[
  \Pr[X=i] = \frac{1}{c(N,t)} \cdot i^{-t} \textrm{ with } c(N,t) = \sum_{j=1}^N j^{-t}.
\]
\begin{wrapfigure}{r}{0.45\textwidth}
    \centering
    \includegraphics[width=0.45\textwidth]{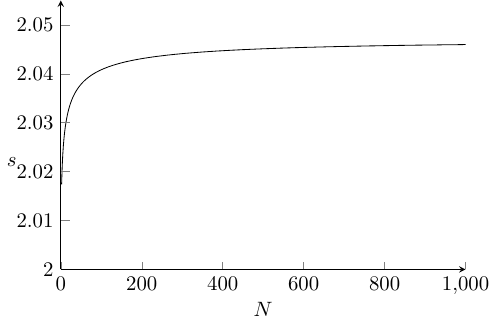}
    \vspace*{-0.8cm}
    \caption{Zipf $\Z(N,0.777)$ Speed-Up }
    \label{fig:zipf}
\end{wrapfigure}
Zipf distributions empirically appear in password database leaks and approximate the observed password distribution. For instance, the well-known LinkedIn database leak with $N=1.6 \cdot 10^8$ passwords can be modelled via $\Z(N, 0.777)$~\cite{malone2012investigating}. Notice that in contrast to all other distributions analyzed in this section, the Zipf distribution is not a product distribution.

Currently, the best known quantum password guessing~\cite{CANS:DGMMS21} achieves a quantum speed-up of (almost) $s=2$ for any Zipf parameter $t$. An application of \Cref{theo:speedup} results in \Cref{fig:zipf}.

For $\Z(N,0.777)$, our analysis improves over~\cite{CANS:DGMMS21} to quantum speed-up $s>2.04$ for password guessing. 

\begin{figure}
\centering
   \subfloat[Geo$(1000,p)$ Speed-up]{
    \includegraphics[width=0.45\textwidth]{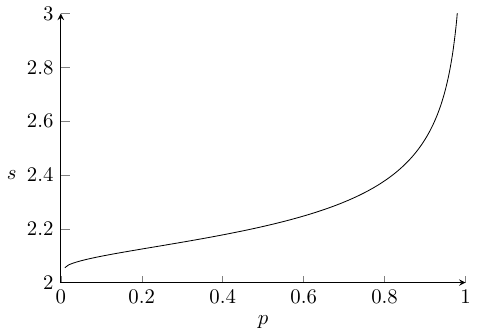}
    \label{fig:geometric}
  }\hfill
  \subfloat[Poi$(1000,\lambda)$ Speed-up]{
    \includegraphics[width=0.45\textwidth]{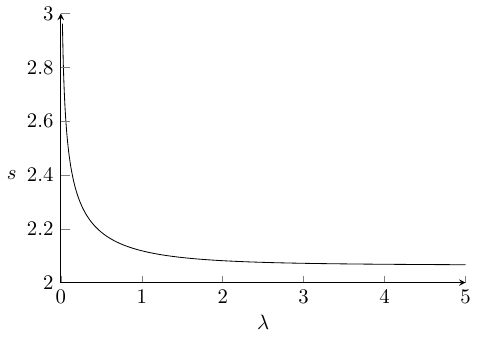}
    \label{fig:poisson}
  }
    \caption{Quantum Speed-ups for Geometric and Poisson distribution}
    
\end{figure}

\medskip

The following two probability distributions, Geometric and Poisson, have to the best of our knowledge not been used for cryptographic key choices. However, due to their wide range of applications in theoretical computer science quantum speed-ups for these distributions might be of independent interest. In order to apply \Cref{theo:speedup}, we define both distributions with finite support.

\paragraph{Geometric.}
The Geometric distribution $\chi=\textrm{Geo}(N,p)$ with $X \sim \operatorname{Geo}(N,p)$ satisfies for $i=0, \ldots , N$
\[
  \Pr[X=i] = \frac{1}{c(N,p)} \cdot (1-p)^i \textrm{ with } c(N,p) = \sum_{j=0}^N(1-p)^j.
\]

Let us consider keys sampled from $(\textrm{Geo}(1000,p))^n$. The quantum speed-ups are depicted in \Cref{fig:geometric}. We observe that we obtain super-quadratic speed-ups for all $p > 0$. When $p$ converges to $1$, our speed-up becomes unbounded.

\paragraph{Poisson.}
The Poisson distribution $\chi = \operatorname{Poi}(N,\lambda)$ with $X \sim \operatorname{Poi}(N,\lambda)$ satisfies for all $i= 0, \ldots , N$
\[
  \Pr[X=i] = \frac 1 {c(N, \lambda)} \cdot \frac{\lambda^k}{k!} \textrm{ with } c(N, \lambda) = \sum_{j=0}^N\frac{\lambda^j}{j!}.
\]
The quantum speed-ups for $\operatorname{Poi}(1000,\lambda)$ are depicted in \Cref{fig:poisson}, with unbounded speed-up for $\lambda \to 0$.

\bibliographystyle{alpha}
\bibliography{abbrev3,crypto,custom}

\newcommand{\etalchar}[1]{$^{#1}$}
\begin{thebibliography}{DGM{\etalchar{+}}21}

\bibitem[ADW01]{ambainis2001average}
Andris Ambainis and Ronald De~Wolf.
\newblock Average-case quantum query complexity.
\newblock {\em Journal of Physics A: Mathematical and General}, 34(35):6741, 2001.

\bibitem[Ale03]{FOCS:Alekhnovich03}
Michael Alekhnovich.
\newblock More on average case vs approximation complexity.
\newblock In {\em 44th FOCS}, pages 298--307. {IEEE} Computer Society Press, October 2003.

\bibitem[Ari96]{arikan1996inequality}
Erdal Arikan.
\newblock An inequality on guessing and its application to sequential decoding.
\newblock {\em IEEE Transactions on Information Theory}, 42(1):99--105, 1996.

\bibitem[AS22]{EPRINT:AlbShe22}
Martin~R. Albrecht and Yixin Shen.
\newblock Quantum augmented dual attack.
\newblock Cryptology ePrint Archive, Report 2022/656, 2022.

\bibitem[BDK{\etalchar{+}}18]{bos2018crystals}
Joppe Bos, L{\'e}o Ducas, Eike Kiltz, Tancr{\`e}de Lepoint, Vadim Lyubashevsky, John~M Schanck, Peter Schwabe, Gregor Seiler, and Damien Stehl{\'e}.
\newblock Crystals-kyber: a cca-secure module-lattice-based kem.
\newblock In {\em 2018 IEEE European Symposium on Security and Privacy (EuroS\&P)}, pages 353--367. IEEE, 2018.

\bibitem[Ber41]{berry1941accuracy}
Andrew~C Berry.
\newblock The accuracy of the gaussian approximation to the sum of independent variates.
\newblock {\em Transactions of the american mathematical society}, 49(1):122--136, 1941.

\bibitem[Ber23]{EPRINT:Bernstein23c}
Daniel~J. Bernstein.
\newblock Asymptotics of hybrid primal lattice attacks.
\newblock Cryptology {ePrint} Archive, Report 2023/1892, 2023.

\bibitem[BHMT02]{brassard2002quantum}
Gilles Brassard, Peter Hoyer, Michele Mosca, and Alain Tapp.
\newblock Quantum amplitude amplification and estimation.
\newblock {\em Contemporary Mathematics}, 305:53--74, 2002.

\bibitem[BKW00]{STOC:BluKalWas00}
Avrim Blum, Adam Kalai, and Hal Wasserman.
\newblock Noise-tolerant learning, the parity problem, and the statistical query model.
\newblock In {\em 32nd ACM STOC}, pages 435--440. {ACM} Press, May 2000.

\bibitem[BM23]{budMar23}
Alessandro Budroni and Erik Mårtensson.
\newblock Improved estimation of key enumeration with applications to solving lwe.
\newblock In {\em 2023 IEEE International Symposium on Information Theory (ISIT)}, pages 495--500, 2023.

\bibitem[CT06]{cover2006elements}
Thomas~M. Cover and Joy~A. Thomas.
\newblock {\em Elements of information theory}.
\newblock John Wiley \& Sons, second edition, 2006.

\bibitem[DGM{\etalchar{+}}21]{CANS:DGMMS21}
Markus D{\"u}rmuth, Maximilian Golla, Philipp Markert, Alexander May, and Lars Schlieper.
\newblock Towards quantum large-scale password guessing on real-world distributions.
\newblock In Mauro Conti, Marc Stevens, and Stephan Krenn, editors, {\em CANS 21}, volume 13099 of {\em {LNCS}}, pages 412--431. Springer, Cham, December 2021.

\bibitem[DP23]{C:DucPul23}
L{\'e}o Ducas and Ludo~N. Pulles.
\newblock Does the dual-sieve attack on learning with errors even work?
\newblock In Helena Handschuh and Anna Lysyanskaya, editors, {\em CRYPTO~2023, Part~III}, volume 14083 of {\em {LNCS}}, pages 37--69. Springer, Cham, August 2023.

\bibitem[Ess45]{esseen1945fourier}
Carl-Gustav Esseen.
\newblock Fourier analysis of distribution functions. a mathematical study of the laplace-gaussian law.
\newblock 1945.

\bibitem[Gro96]{Grover96}
Lov~K. Grover.
\newblock A fast quantum mechanical algorithm for database search.
\newblock In Gary~L. Miller, editor, {\em Proceedings of the Twenty-Eighth Annual {ACM} Symposium on the Theory of Computing, Philadelphia, Pennsylvania, USA, May 22-24, 1996}, pages 212--219. {ACM}, 1996.

\bibitem[HB01]{AC:HopBlu01}
Nicholas~J. Hopper and Manuel Blum.
\newblock Secure human identification protocols.
\newblock In Colin Boyd, editor, {\em ASIACRYPT~2001}, volume 2248 of {\em {LNCS}}, pages 52--66. Springer, Berlin, Heidelberg, December 2001.

\bibitem[H{\o}y00]{hoyer2000arbitrary}
Peter H{\o}yer.
\newblock Arbitrary phases in quantum amplitude amplification.
\newblock {\em Physical Review A}, 62(5):052304, 2000.

\bibitem[HPS98]{HofPipSil98}
Jeffrey Hoffstein, Jill Pipher, and Joseph~H. Silverman.
\newblock {NTRU:} {A} ring-based public key cryptosystem.
\newblock In {\em Third Algorithmic Number Theory Symposium (ANTS)}, volume 1423 of {\em {LNCS}}, pages 267--288. Springer, June 1998.

\bibitem[HRSS17]{CHES:HRSS17}
Andreas H{\"u}lsing, Joost Rijneveld, John~M. Schanck, and Peter Schwabe.
\newblock High-speed key encapsulation from {NTRU}.
\newblock In Wieland Fischer and Naofumi Homma, editors, {\em CHES~2017}, volume 10529 of {\em {LNCS}}, pages 232--252. Springer, Cham, September 2017.

\bibitem[IDF22]{matzov}
MATZOV IDF.
\newblock Report on the security of lwe:improved dual lattice attack, 2022.
\newblock \url{https://zenodo.org/record/6412487#.ZCrT7-xBxqs}.

\bibitem[KKNM25]{ac25hybrid}
Alexander Karenin, Elena Kirshanova, Julian Nowakowski, and Alexander May.
\newblock {Fast Slicer for Batch-CVP: Making Lattice Hybrid Attacks Practical}.
\newblock In {\em International Conference on the Theory and Application of Cryptology and Information Security (ASIACRYPT), to appear}, 2025.

\bibitem[MM12]{malone2012investigating}
David Malone and Kevin Maher.
\newblock Investigating the distribution of password choices.
\newblock In {\em Proceedings of the 21st international conference on World Wide Web}, pages 301--310, 2012.

\bibitem[Mon11]{montanaro2011quantum}
Ashley Montanaro.
\newblock Quantum search with advice.
\newblock In {\em Theory of Quantum Computation, Communication, and Cryptography: 5th Conference, TQC 2010, Leeds, UK, April 13-15, 2010, Revised Selected Papers 5}, pages 77--93. Springer, 2011.

\bibitem[Reg03]{STOC:Regev03}
Oded Regev.
\newblock New lattice based cryptographic constructions.
\newblock In {\em 35th ACM STOC}, pages 407--416. {ACM} Press, June 2003.

\end{thebibliography}


\end{document}